\documentclass[letterpaper, 10 pt,journal]{IEEEtran}  

\usepackage{graphicx}
\usepackage{amsmath,amssymb}   
\usepackage{tikz}
\usepackage{cite}
\usepackage{subfigure}

\usepackage{amsthm}
\theoremstyle{remark}
\newtheorem{thm}{Theorem}
\newtheorem{defn}[thm]{Definition}
\newtheorem{assu}[thm]{Assumption}
\newtheorem{algo}[thm]{Algorithm}

\newtheorem{rmk}[thm]{Remark}

\newcommand{\overlength}[1]{#1}  				
\renewcommand{\overlength}[1]{}  				
\newcommand{\overlengthSeven}[1]{#1} 			
\newcommand{\Six}[1]{#1}				        
\newcommand{\DeleteSix}[1]{#1}					
\renewcommand{\DeleteSix}[1]{}  				

\Six{\linespread{0.99}}

\title{\LARGE \bf
Iterative data-driven inference of nonlinearity measures\\ via successive graph approximation
}

\author{Tim Martin and Frank Allg{\"o}wer*
	\thanks{*T. Martin and F. Allg{\"o}wer are with the Institute for Systems Theory and Automatic Control, University of Stuttgart. This work was funded by Deutsche Forschungsgemeinschaft (DFG, German Research Foundation) under Germany's Excellence Strategy - EXC 2075 - 390740016. For correspondence, mailto: tim.martin@ist.uni-stuttgart.de.}
}

\begin{document}

\IEEEoverridecommandlockouts

\IEEEpubid{\begin{minipage}{\textwidth}\ \\[12pt] \copyright 2020 IEEE. This version has been accepted for publication in Proc. Conference on Decision and Control, 2020. Personal use of this material is permitted. Permissionfrom EUCA must be obtained for all other uses, in any current or future media, including reprinting/republishing this material for advertising or promotionalpurposes, creating new collective works, for resale or redistribution to servers or lists, or reuse of any copyrighted component of this work in other works.\end{minipage}}

\maketitle
\pagestyle{empty}

\begin{abstract}

In this paper, we establish an iterative data-driven approach to derive guaranteed bounds on nonlinearity measures of unknown nonlinear systems. In this context, nonlinearity measures quantify the strength of the nonlinearity of a dynamical system by the distance of its input-output behaviour to a set of linear models. First, we compute a guaranteed upper bound of these measures by given input-output samples based on a data-based non-parametric set-membership representation of the ground-truth system and local inferences of nonlinearity measures. Second, we propose an algorithm to improve this bound iteratively by further samples of the unknown input-output behaviour. 

\end{abstract}

\section{Introduction}\label{Intrduction}

Deriving a controller for complex systems requires usually a sufficiently precise model. However, modelling such systems is difficult and more time consuming than the controller design. For this purpose, data-driven controller design, where a controller is obtained without identifying a model, has been investigated. \cite{c12} gives an overview of such approaches.\\\indent 
One data-driven approach is examined in \cite{c6} where control-theoretic system properties, as $\mathcal{L}_2$-gain and conic relations, are learned from given input-output samples. These properties give insight to the open-loop system and facilitate the application of well-known feedback theorems. Analogously, \cite{c100} deduces from given input-output samples a linear surrogate model that minimizes the maximal deviation to the unidentified nonlinear system. By the knowledge of the linear model and its approximation error, techniques from robust control theory can be applied to determine a controller with closed-loop guarantees. Furthermore, the approximation error corresponds to the nonlinearity measures from \cite{c3} and \cite{c2} which quantify the nonlinearity of dynamical systems.\\\indent
The drawback of the approach from \cite{c6} is the requirement of a large number of input-output samples. To this end, iterative approaches are investigated where the control-theoretic properties is identified by performing sequentially experiments on the plant. These algorithms provide an (optimal) decision what experiment should be applied next to improve the estimation of the system property. For example, the $\mathcal{L}_2$-gain and a linear surrogate model are computed in \cite{c8} and \cite{c101}, respectively, for linear time-invariant (LTI) systems based on solving optimization problems using gradient-based methods. \\\indent  
To deduce an iterative scheme for nonlinear systems, we exploit in this work a non-parametric data-based model.\overlength{ The most famous non-parametric modelling technique is Gaussian process regression \cite{c112}. There, the set of possible functions is described by means of a Gaussian distribution, where its mean function and covariance kernel is updated with additional data to achieve a more precise estimation. Gaussian process regression is successfully applied in system identification, control, passivity verification \cite{c109}, and optimization with unknown objective function \cite{c111}.} Instead of a statistical approach\overlengthSeven{ as in  Gaussian process regression \cite{c112}}, we study a non-parametric set-membership model where an envelope of the graph of a Lipschitz function is described directly from given input-output data. Such Lipschitz approximations are investigated, e.g., in set-membership identification \cite{c108} and in Kinky inferences for nonlinear model predictive control \cite{c102,c113}. \\\indent
In this paper, we determine a guaranteed upper bound on nonlinearity measures by means of an envelope that contains the unknown input-output behaviour of the ground-truth system. Especially, the bound is obtained from the maximal distance of a linear approximation model to all realizations of mappings which are contained in the envelope. A bound on this distance is calculated by local inferences of the nonlinearity measure. Thus, this approach constitutes an alternative to \cite{c6} and \cite{c100} for deriving guaranteed bounds on system-theoretic properties from given input-output samples. Moreover, contrary to \cite{c6} and \cite{c100}, the computation of the covering radius is not required and noisy output measurements can be considered. Furthermore, we extend this approach to an iterative scheme based on a branch-and-bound algorithm \cite{BnB} to reduced the derived upper bound on the nonlinearity measure by further sampling. Here, we ensure that the computational complexity of this algorithm does not increase with further iterations and prove the convergence to the true nonlinearity measure in absence of noise. \\\indent 
In contrast to set-membership identification, we approximate the nonlinear system by a linear model that is in general not contained in the envelope. Instead, the linear model is a projection of the envelope on a set of linear models, and therefore the results of \cite{c108} are not applicable. Moreover, our goal differs from system identification in the sense that we approximate the behaviour of a complex (nonlinear) system by a simple (linear) model, i.e., the linear model is a surrogate model of the nonlinear system.\\\indent 
The paper is organized as follows. First, we introduce nonlinearity measures and specify the problem of estimating nonlinearity measures via graph approximation. Then we solve this problem by means of local inference of nonlinearity measures. Subsequent, we propose the iterative scheme to improve the estimation of the nonlinearity measure by successive sampling. The paper concludes with \DeleteSix{some extensions of the iterative scheme, e.g., to increase its convergence rate, and }a numerical example where the iterative scheme is compared to the offline approach \cite{c100}.

 \IEEEpubidadjcol

\section{Problem setup and Definition of a nonlinearity measure}\label{PreSet}

Let the input-output behaviour of the unknown discrete-time nonlinear SISO system be described by the mapping
\begin{equation*}
	N:\mathcal{U}\subset\mathbb{R}^n\rightarrow\mathcal{Y}\subset\mathbb{R}^n,
\end{equation*}
i.e. $N$ maps input on output trajectories of length $n$. We suppose that $N(0)=0$ without loss of generality.\\\indent
Moreover, let the input set $\mathcal{U}$ be spanned by an orthonormal basis of signals $v_i\in\mathbb{R}^n,i=1,\dots,\mu\leq n$
\begin{equation}\label{InputSet}
	\mathcal{U} = \{u\in\mathbb{R}^n: u=\begin{bmatrix}v_1 & \cdots & v_\mu\end{bmatrix}\bar{u}, \bar{u}\in\bar{\mathcal{U}}\subset\mathbb{R}^\mu\}
\end{equation}
where the amplitudes $\bar{u}$ are bounded by the box constraint $\bar{u}\in\bar{\mathcal{U}}=[\underline{\alpha}_1,\bar{\alpha}_1]\times\dots\times[\underline{\alpha}_\mu,\bar{\alpha}_\mu]$. This compact input set is also assumed in \cite{c6} and \cite{c100} since it is often considered in system identification where the basis $v_1,\dots,v_\mu$ is chosen, e.g., to a Fourier basis or Legendre polynomials. Moreover, a suitable choice of basis signals ensures that all inputs which are suggested by our iterative scheme are experimentally admissible. Note that each input $u\in\mathcal{U}$ corresponds to an unique amplitude $\bar{u}\in\bar{\mathcal{U}}$ because $v_1\dots,v_\mu$ is a orthogonal basis and $\mu\leq n$. Therefore, we can exchange $u$ by its corresponding amplitude $\bar{u}$ and vice versa.\\\indent
Furthermore, we suppose that $\mathcal{Y}$ is a compact set to ensure the well-definiteness of the following definition of nonlinearity measures from \cite{c3}.
\begin{defn}[AE-NLM]\label{NLMDef}
The nonlinearity of a dynamical system $N:\mathcal{U}\subset\mathbb{R}^n\rightarrow\mathcal{Y}\subset\mathbb{R}^n$ is quantified by the additive error nonlinearity measure (AE-NLM)
\begin{align}\label{AENLM}
	\Phi^{\mathcal{U},\mathcal{G}}_{\text{AE}}&:=\inf_{G\in\mathcal{G}}\sup_{u\in\mathcal{U}\backslash \{0\}}\frac{||N(u)-G(u)||}{||u||}
\end{align}
where $||\cdot||$ denotes the Euclidean vector norm and $G:\mathcal{U}\rightarrow\mathcal{Y}$ is an element of a set $\mathcal{G}$ of LTI systems.
\end{defn}
Solving the optimization problem~\eqref{AENLM} yields the\ \textquoteleft best\textquoteright\ linear approximation $G^*$ that minimizes the gain of the error system $\Delta:=N-G^*$ with respect to the Euclidean norm. Thus, the nonlinear system can be written as the interconnection of the linear model $G^*$ and the error model $\Delta$ which gain corresponds to the AE-NLM. Therefore, techniques from robust control theory can be applied once $G^*$ and the nonlinearity measure are known. Furthermore, $G^*$ can be seen as the projection of the nonlinear system $N$ on the set of linear systems $\mathcal{G}$ as illustrated in Figure~\ref{Fig:NLMProject}.
\begin{figure}
	\begin{center}
		\begin{tikzpicture}[scale=0.6]
\draw  plot[smooth, tension=.7] coordinates {(-3.5,0.5) (-3,2.5) (-1,3) (1.5,2.3) (3.5,2.5) (3.75,1.5) (4,0) (2.5,0.5) (0,-0.25) (-3,-0.5) (-3.5,0.5)};

\path (-2.5,0.5) node (G_set) {$\mathcal{G}$}; 
\path (0,0.5) node[left,circle, scale=0.3, fill, draw, black] (Zero) {}; 
\path (Zero) node[left] (Zero1) {$0$}; 
\path (1.8,1.0) node[below, circle, scale=0.3, fill, draw, black] (G_star) {}; 
\path (G_star) node[right] (G_star1) {$G^*$};
\path (1.8,3) node[above, circle, scale=0.3, fill, draw, black] (N) {};
\path (N) node[above] (N1) {$N$}; 
\path (0.1,2.1) node (Norm_N) {$||N||$}; 
\path (2.5,1.9) node (Phi) {$\Phi^{\mathcal{U},\mathcal{G}}_{\text{AE}}$};

\draw[-,line width=1pt] (Zero) -- (G_star);
\draw[-,line width=1pt] (G_star) -- (N);
\draw[-,line width=1pt] (N) -- (Zero);

\end{tikzpicture}
	\end{center}
	\caption{Projection of a nonlinear system $N$ on the set of linear models $\mathcal{G}$.}
	\label{Fig:NLMProject}
\end{figure}
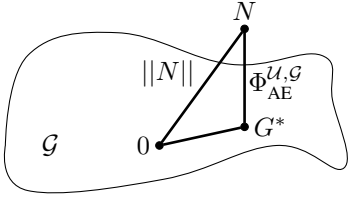
As shown in \cite{c100}, the AE-NLM is related to the $\ell_2$-gain and the conic relations from \cite{c5} by the special choice $\mathcal{G}=\{0\}$ and $\mathcal{G}=\{G=cI:c\in\mathbb{R}\}$ with $I:u\mapsto u$, respectively. For further reading on nonlinearity measures, we refer to \cite{c100} where, amongst other things, parametrizations of $\mathcal{G}$ are proposed and a characterization of stability for feedback interconnections using nonlinearity measures is derived via the concept of graph separation. 

\section{A data-based non-parametric model for Lipschitz mappings}

In this section, we introduce the data-based non-parametric model for Lipschitz mappings from \cite{c108}. Subsequent, we specify the problem setup to calculate an upper bound for the AE-NLM from this model.\\\indent 
To conclude on the input-output behaviour of the unknown nonlinear system $N$, we assume the access of input-output trajectories
\begin{equation*}
	\mathcal{U}_D\times\mathcal{Y}_D:=\{(u_1,y_1),\dots,(u_D,y_D)\}\subset\mathcal{U}\times\mathcal{Y}
\end{equation*} 
of the nonlinear system, i.e., $y_i=N(u_i),i=1,\dots,D$. Moreover, let $\bar{\mathcal{U}}_D$ denote the set of amplitudes $\bar{u}_1,\dots,\bar{u}_D$ that correspond to $u_1,\dots,u_D$. Since the set of mappings generating $\mathcal{U}_D\times\mathcal{Y}_D$ is unbounded, the rate of variation of $N$ is restricted as in \cite{c108}.
\begin{assu}[Lipschitz-continuity]\label{AssuLipschitz}
	The input-output mapping $N:\mathcal{U}\rightarrow\mathcal{Y}$ is Lipschitz continuous with $L>0$, i.e.,
\begin{equation*}\label{LipschitzIneq}
	||N(u)-N(u')||\leq L||u-u'||,\quad \forall u,u'\in\mathcal{U}
\end{equation*}
and the Lipschitz constant $L$ is known.
\end{assu}
In general, the Lipschitz constant $L$ is not known beforehand. However, different data-driven methods were developed to estimate $L$, e.g., Strongin’s estimator \cite{c104} and POKI \cite{c105}. Under the prior knowledge of $\mathcal{U}_D\times\mathcal{Y}_D$ and Assumption~\ref{AssuLipschitz}, we can conclude that the graph of the mapping $N$ is contained in the envelope
\begin{align*}
	E(\mathcal{U}_D\times\mathcal{Y}_D):=& \{(u,y)\in\mathcal{U}\times\mathcal{Y}: ||y-y_i||\leq L||u-u_i||,\\
	 &\hspace{4.0cm}i=1,\dots,D\} .
\end{align*}
An illustration of this envelope can be found in Figure~\ref{Fig:envelope} for $n=1$ and in \cite{c108} for higher dimensions. 
\begin{figure}
	\begin{center}
		\includegraphics[width=0.9\linewidth]{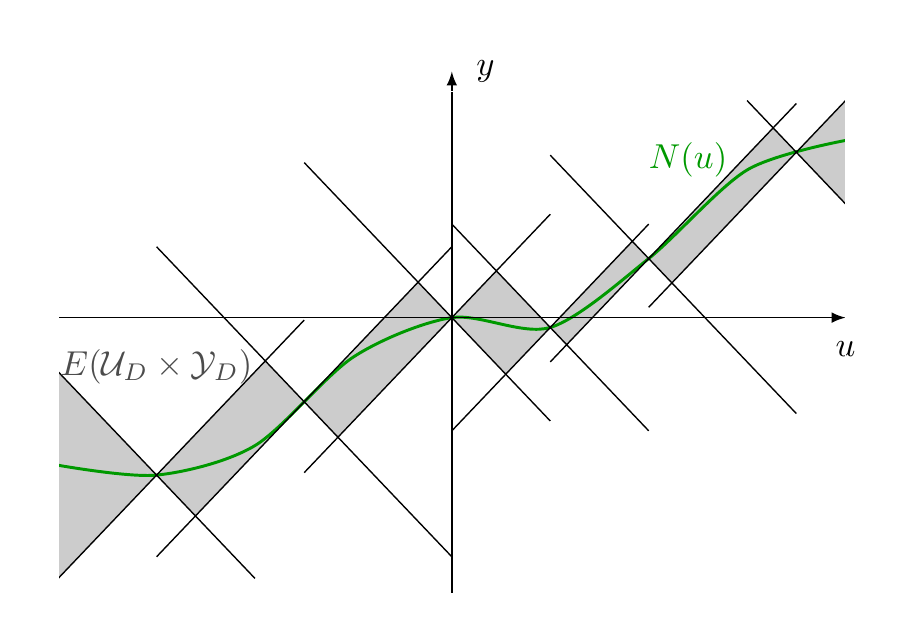}
	\end{center}
	\caption{Envelope of an one dimensional input-output mapping $y=N(u)$.}
	\label{Fig:envelope}
\end{figure}
Since $E(\mathcal{U}_D\times\mathcal{Y}_D)$ is defined through the input-output samples $\mathcal{U}_D\times\mathcal{Y}_D$, the envelope establishes a data-based non-parametric set-membership representation of the ground-truth mapping $N$.\\\indent
In the remainder of this paper, we exploit the envelope $E(\mathcal{U}_D\times\mathcal{Y}_D)$ to determine an upper bound on the AE-NLM. Since the graph of $N$ is a subset of $E(\mathcal{U}_D\times\mathcal{Y}_D)$,
\begin{equation}\label{Upperbound}
\begin{aligned}
	&\inf_{G\in\mathcal{G}}\max_{u\in\mathcal{U},||u||\geq\epsilon}\frac{||N(u)-G(u)||}{||u||}\\
	\leq&\inf_{G\in\mathcal{G}}\max_{\substack{(u,y)\in E(\mathcal{U}_D\times\mathcal{Y}_D),\\||u||\geq\epsilon}}\frac{||y-G(u)||}{||u||}{=:}\Phi^{E(\mathcal{U}_D\times\mathcal{Y}_D),\mathcal{G}}_{\text{AE}}.
\end{aligned}
\end{equation}
First, observe that we exclude small inputs $||u||<\epsilon$ in \eqref{Upperbound} similar to \cite{c6} and \cite{c100} as otherwise the upper bound $\Phi^{E(\mathcal{U}_D\times\mathcal{Y}_D),\mathcal{G}}_{\text{AE}}$ would be at least $L$ regardless of the data set $\mathcal{U}_D\times\mathcal{Y}_D$. Indeed, there exists a neighbourhood of $(u,y)=(0,0)$ with $E(\{(0,0)\})=E(\mathcal{U}_D\times\mathcal{Y}_D)$ and
\begin{equation*}
	\inf_{G\in\mathcal{G}}\sup_{\substack{(u,y)\in E(\{(0,0)\}),\\ u\neq0}}\frac{||y-G(u)||}{||u||}{=} \sup_{\substack{(u,y)\in E(\{(0,0)\}),\\ u\neq0}}\frac{||y||}{||u||}{=}L.
\end{equation*}
Here, the first equality holds due to the optimal approximation of Lipschitz functions from \cite{c108}.\overlength{ Instead of excluding small inputs in \eqref{Upperbound}, we could extend the definition of nonlinearity measures by a fixed constant $\epsilon>0$
\begin{align*}
	\Phi^{\mathcal{U},\mathcal{G}}_{\text{AE}}=&\inf_{G\in\mathcal{G}}\sup_{u\in\mathcal{U}\backslash\{0\}}\frac{||N(u)-G(u)||}{||u||}+\epsilon.
\end{align*}
Such definition is common for quadratic constraints for nonlinear systems to consider initial conditions.}\\\indent
Second, note that there always exist inputs that solves the left- and right-hand side of \eqref{Upperbound} as the input set is compact by assumption. However, the solution is not necessarily unique.
\overlength{
\begin{rmk}\label{Sprocedure}
To solve the optimization problem of $\Phi^{E(\mathcal{U}_D\times\mathcal{Y}_D),\mathcal{G}}_{\text{AE}}$ in \eqref{Upperbound}, we first follow \cite{c103}. Thus, the application of the S-procedure yields for the right-hand side of \eqref{Upperbound}
\begin{equation}\label{Opt3}
\begin{aligned}
	&\inf_{G\in\mathcal{G},t\geq0,\lambda_1,\dots,\lambda_D\geq0} t\\
	&\text{s.t.}\ \begin{bmatrix}\Psi^TM^TM\Psi-\begin{bmatrix}tI & 0\\ 0 & 0\end{bmatrix} & 0\\ 0 & 0\end{bmatrix}{-}\sum_{i=1}^{D}\lambda_iQ_L^N(u_i){\leq}0
\end{aligned}
\end{equation}
with 
\begin{align*}
Q_L^N(u_i)&{=}\begin{bmatrix}-L^2I & 0 & L^2u_i\\ 0 & I & -y_i\\ L^2{u_i}^T & -y_i^T & y_i^Ty_i-L^2{u_i}^Tu_i\end{bmatrix}, i{=}1,\dots,D,\\
\Psi&{=}\begin{bmatrix}
G & 0 \\ 0 & I \end{bmatrix}\ \text{and}\ M=\frac{1}{\sqrt{2}}\begin{bmatrix}
I & -I \\ -I & I \end{bmatrix}.
\end{align*}
For the sake of simplicity, we omit in \eqref{Opt3} the bound of the input set $\mathcal{U}$ and the exclusion of small inputs by, e.g., ellipsoids. In \eqref{Opt3}, taking the infimum of $t$ leads to the optimal linear model $G^*=0$ as the linear model $G$ only appears in the positive semi-definite matrix $\Psi^TM^TM\Psi$. The problem is that the S-procedure is only a sufficient condition. Moreover, the relaxation of the S-procedure provides a conservative approximation of the nonlinearity measure, even if the optimization over a set of linear models is omitted, and other quadratic system properties, e.g., passivity. Hence,the S-procedure exhibits a data-inefficient estimation for quadratic system properties.  
\end{rmk}}
\overlengthSeven{
\begin{rmk}\label{Sprocedure}
To solve the optimization problem of $\Phi^{E(\mathcal{U}_D\times\mathcal{Y}_D),\mathcal{G}}_{\text{AE}}$ in \eqref{Upperbound}, we could pursue \cite{c103} by applying the S-procedure to derive a semi-definite programming for the right-hand side of \eqref{Upperbound}. However, due to the relaxation of the S-procedure, the problem emerges that the optimal linear model is zero regardless of the data set and the relaxation exhibits a data-inefficient estimation for \eqref{Upperbound} as well as for other quadratic system properties as, e.g., passivity. 
\end{rmk}}

\section{Data-driven inference of the nonlinearity measure}\label{IterativeApproach}

In this section, we present our two contributions. First, we establish an approach to solve the optimization problem of $\Phi^{E(\mathcal{U}_D\times\mathcal{Y}_D),\mathcal{G}}_{\text{AE}}$ for a given linear approximation model and given input-output samples. Subsequent, we extend this approach to an iterative procedure based on a branch-and-bound algorithm \cite{BnB} to improve the upper bound of the AE-NLM by iterative sampling.\DeleteSix{ In Section~\ref{ExtenScheme}, we present some modifications of the iterative scheme. }

\subsection{Inference of the AE-NLM by local inferences}\label{Relaxation}

In this section, we solve for a given linear surrogate model $G$ the optimization problem 
\begin{align}\label{UpperboundFixG}
	\Phi^{E(\mathcal{U}_D\times\mathcal{Y}_D),G}_{\text{AE}}=\max_{\substack{(u,y)\in E(\mathcal{U}_D\times\mathcal{Y}_D),\\||u||\geq\epsilon}}\frac{||y-G(u)||}{||u||},
\end{align}
which corresponds to the right-hand side of \eqref{Upperbound} with $\mathcal{G}=\{G\}$. Similar to \cite{c113}, we obtain the global inference \eqref{UpperboundFixG} by means of local inferences of the AE-NLM. To properly calculate the local AE-NLM inferences, we also
suggest \DeleteSix{two nonconvex relaxations}\Six{a nonconvex relaxation}.\\\indent
According to localised Kinky inference from \cite{c113}, we consider a partition of $\bar{\mathcal{U}}$ by hyperrectangles $\bar{\mathcal{U}}_{H_1},\dots,\bar{\mathcal{U}}_{H_h}$ and its resulting partition $\mathcal{U}_{H_1},\dots,\mathcal{U}_{H_h}$ of $\mathcal{U}$. Since the partition covers the whole input set, the solution of \eqref{UpperboundFixG} is obtained by solving
\begin{align}\label{SubProb}
	\Phi_{\text{AE}}^{E(\mathcal{U}_D\times\mathcal{Y}_D),\mathcal{U}_{H_i},G}=\max_{\substack{(u,y)\in E(\mathcal{U}_D\times\mathcal{Y}_D),\\u\in \mathcal{U}_{H_i},||u||\geq\epsilon}}\frac{||y-G(u)||}{||u||}
\end{align}
for each subset $\mathcal{U}_{H_i},i=1,\dots,h$ and then by taking the maximum over all $\Phi_{\text{AE}}^{E(\mathcal{U}_D\times\mathcal{Y}_D),\mathcal{U}_{H_i},G},i=1,\dots,h$.\\\indent
Due to the increase of constraints with the number of samples of the envelope $E(\mathcal{U}_D\times\mathcal{Y}_D)$ in the optimization problem \eqref{SubProb}, we exploit the notion of local inference of the AE-NLM analogously to localised Kinky inference. 
\begin{defn}[Local inference of AE-NLM]
For each subset ${\mathcal{U}}_{H_1},\dots,{\mathcal{U}}_{H_h}$, we define the local AE-NLM inference 
\begin{align}\label{LocalOpt}
	\max_{\substack{(u,y)\in E(\{(u_{H_i}',y_{H_i}'),(u_{H_i}'',y_{H_i}'')\}),\\u\in \mathcal{U}_{H_i},||u||\geq\epsilon}}\frac{||y-G(u)||}{||u||}
\end{align}
where the two samples $(u_{H_i}',y_{H_i}'),(u_{H_i}'',y_{H_i}'')\in\mathcal{U}_D\times\mathcal{Y}_D$ are chosen such that the inputs $u_{H_i}'$ and $u_{H_i}''$ are the closest samples of ${\mathcal{U}}_D$ to ${\mathcal{U}}_{H_i}$.
\end{defn}
Note that the local inference~\eqref{LocalOpt} is only an upper bound of the global inference~\eqref{SubProb} as $E(\mathcal{U}_D\times\mathcal{Y}_D)\subseteq E(\{(u_{H_i}',y_{H_i}'),(u_{H_i}'',y_{H_i}'')\})$. However, the number of constraints in \eqref{LocalOpt} is reduced significant compared to \eqref{SubProb} and regardless of the number of samples in $\mathcal{U}_D$. Furthermore, the consideration of the two closest data samples for each subset ${\mathcal{U}}_{H_i}$ in \eqref{LocalOpt} is reasonable as these samples mostly generate actives constraints in \eqref{SubProb}. This choice of samples is also motivated by the case $n=1$ where 
the global inference \eqref{SubProb} and the local inference \eqref{LocalOpt} are equivalent.\\\indent 
In the following \DeleteSix{two theorems, we present nonconvex relaxations}\Six{theorem, we present a nonconvex relaxation} based on geometrical arguments to further reduce the computationally complexity of the local inference \eqref{LocalOpt}.
\begin{thm}\label{Relaxation1}
Let two input-output samples $\mathcal{U}_D^\ell\times\mathcal{Y}_D^\ell=\{(u_1,y_1),(u_2,y_2)\}$ be given. Then, the local inference of the AE-NLM \eqref{LocalOpt} is bounded from above by
\begin{equation}\label{NonconvexRelaxation1}
\begin{aligned}
	\max\{\alpha_{\text{AE}}^{E(\mathcal{U}_D^\ell\times\mathcal{Y}_D^\ell),G},\beta_{\text{AE}}^{E(\mathcal{U}_D^\ell\times\mathcal{Y}_D^\ell),G},\gamma_{\text{AE}}^{E(\mathcal{U}_D^\ell\times\mathcal{Y}_D^\ell),G}\}
\end{aligned}
\end{equation} 	
with
\begin{align*}
	\alpha_{\text{AE}}^{E(\mathcal{U}_D^\ell\times\mathcal{Y}_D^\ell),G}&{=}\max_{\substack{u\in \mathcal{U}_{H_i},||u||\geq\epsilon,\\r_2^2\leq r_1^2+r^2, r_1^2\leq r_2^2+r^2 }}\frac{||M(u)-G(u)||+d(u)}{||u||},\\
	\beta_{\text{AE}}^{E(\mathcal{U}_D^\ell\times\mathcal{Y}_D^\ell),G}&{=}\max_{\substack{u\in \mathcal{U}_{H_i},||u||\geq\epsilon,\\r_2^2>r_1^2+r^2}}\frac{||y_1-G(u)||+r_1(u)}{||u||},\\
	\gamma_{\text{AE}}^{E(\mathcal{U}_D^\ell\times\mathcal{Y}_D^\ell),G}&{=}\max_{\substack{u\in \mathcal{U}_{H_i},||u||\geq\epsilon,\\r_1^2>r_2^2+r^2}}\frac{||y_2-G(u)||+r_2(u)}{||u||},	
\end{align*}	
and the geometric variables
\begin{align*}
	r=||&y_1-y_2||,	r_1(u)=L||u-u_1||,	r_2(u)=L||u-u_2||,\\
	d(u)&=\frac{1}{2r}\sqrt{(r^2-(r_2-r_1)^2)((r_1+r_2)^2-r^2)},\\
	M(u)&=y_2+\sqrt{L^2||u-u_2||^2-d(u)^2}\frac{y_1-y_2}{||y_1-y_2||}.	
\end{align*}
\end{thm} 
\begin{proof}
Let $\mathcal{Y}_{E(\mathcal{U}_D^\ell\times\mathcal{Y}_D^\ell)}(u)$ denote the projection of $E(\mathcal{U}_D^\ell\times\mathcal{Y}_D^\ell)$ for an input $u\in\mathcal{U}$ on $\mathcal{Y}$, i.e., the set of possible outputs for input $u$ included in $E(\mathcal{U}_D^\ell\times\mathcal{Y}_D^\ell)$. Since $\mathcal{U}_D^\ell\times\mathcal{Y}_D^\ell$ contains two data samples, $\mathcal{Y}_{E(\mathcal{U}_D^\ell\times\mathcal{Y}_D^\ell)}(u)$ corresponds to the intersection of two $n-1$-dimensional spheres with center $y_1$ and $y_2$, respectively, and radius $r_1(u)$ and $r_2(u)$, respectively. Due to the Lipschitz continuity of $N$, $\mathcal{Y}_{E(\mathcal{U}_D^\ell\times\mathcal{Y}_D^\ell)}(u)$ is non-empty and contains a $n-2$ dimensional sphere with diameter $2d(u)$ and center $M(u)$. To derive the upper bound \eqref{NonconvexRelaxation1}, we bound the distance of $G(u)$ to all outputs in $\mathcal{Y}_{E(\mathcal{U}_D^\ell\times\mathcal{Y}_D^\ell)}(u)$. To this end, we distinguish between three possible cases depending on the location of the center $M(u)$ and then take the maximum over those cases as in \eqref{NonconvexRelaxation1}.\\\indent
In the first case, the center $M(u)$ lies between $y_1$ and $y_2$ as depicted in Figure~\ref{Fig:ProofCase1}.\DeleteSix{
\begin{figure}
	\begin{center}
		\includegraphics[width=0.85\linewidth]{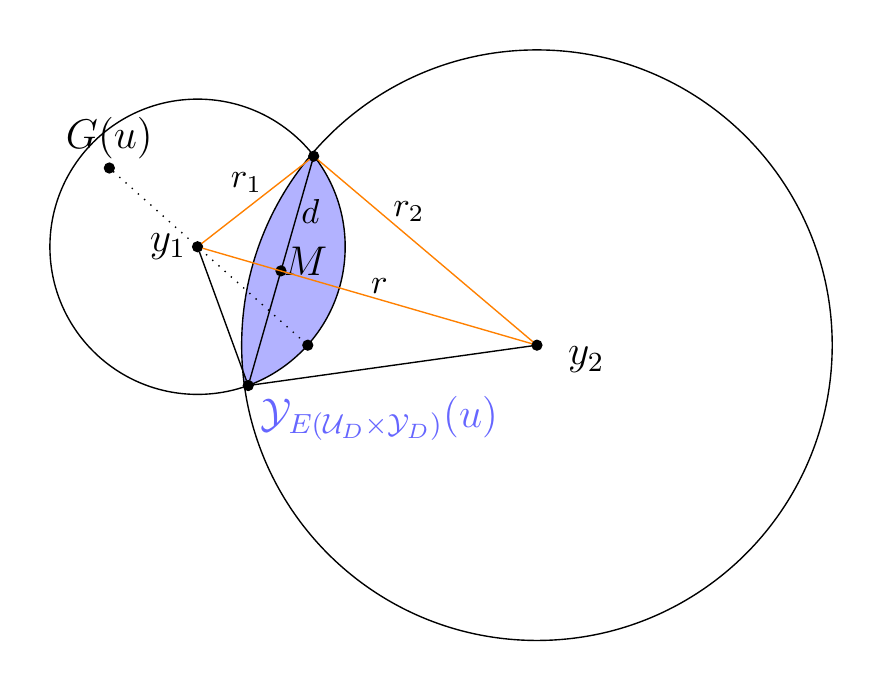}
	\end{center}
	\caption{Set of possible outputs $\mathcal{Y}_{E(\mathcal{U}_D^\ell\times\mathcal{Y}_D^\ell)}(u)$ by Lipschitz continuity at two samples. The case with center $M(u)$ between the output samples $y_1$ and $y_2$.}
	\label{Fig:ProofCase1}
\end{figure}}\Six{
\begin{figure}
	\begin{center}
		\includegraphics[width=0.68\linewidth]{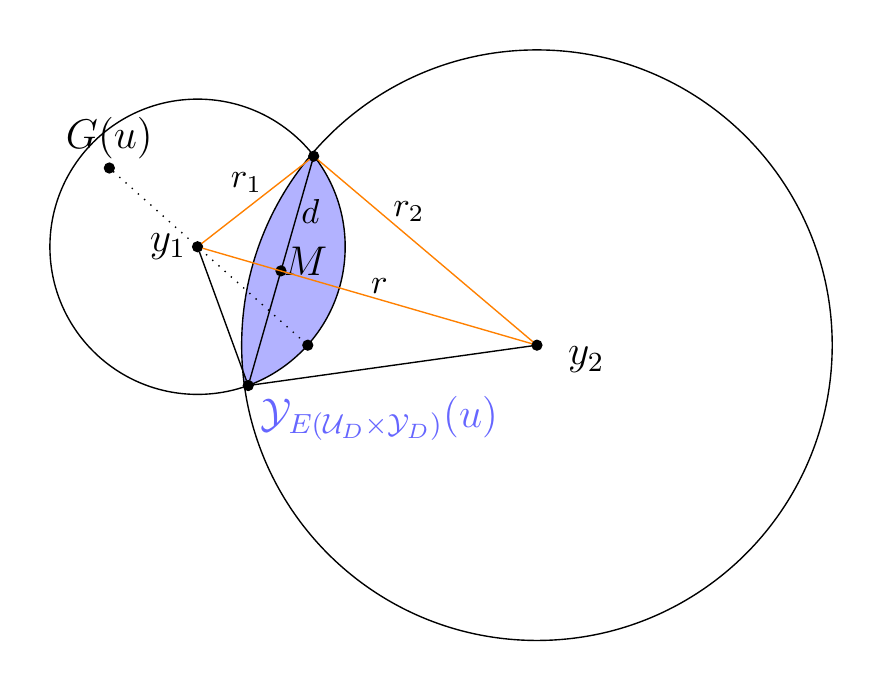}
	\end{center}
	\caption{Set of possible outputs $\mathcal{Y}_{E(\mathcal{U}_D^\ell\times\mathcal{Y}_D^\ell)}(u)$ by Lipschitz continuity at two samples. The case with center $M(u)$ between the output samples $y_1$ and $y_2$.}
	\label{Fig:ProofCase1}
\end{figure}
}
Since the distance of $M(u)$ to any point in $\mathcal{Y}_{E(\mathcal{U}_D\times\mathcal{Y}_D)}(u)$ is less than or equal to the half of the diameter $2d(u)$, the triangle inequality yields 
\begin{equation*}
	\max_{y\in\mathcal{Y}_{E(\mathcal{U}_D^\ell\times\mathcal{Y}_D^\ell)}(u)}||y-G(u)||\leq ||M(u)-G(u)||+d(u)
\end{equation*}
which corresponds to $\alpha_{\text{AE}}^{E(\mathcal{U}_D^\ell\times\mathcal{Y}_D^\ell),G}$. \\\indent
The second case is depicted in Figure~\ref{Fig:ProofCase2} where the center $M(u)$ doesn't lie between $y_1$ and $y_2$ or even one sphere is completely included in the other.\DeleteSix{
\begin{figure}
	\begin{center}
		\includegraphics[width=0.7\linewidth]{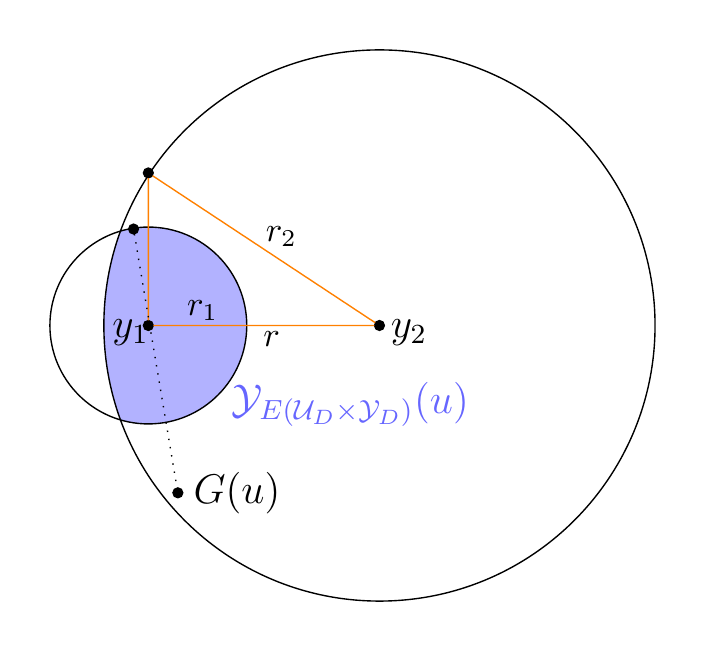}
	\end{center}
	\caption{Set of possible outputs $\mathcal{Y}_{E(\mathcal{U}_D^\ell\times\mathcal{Y}_D^\ell)}(u)$ by Lipschitz continuity at two samples. The case with center $M(u)$ not between the output samples $y_1$ and $y_2$.}
	\label{Fig:ProofCase2}
\end{figure}}\Six{
\begin{figure}
	\begin{center}
		\includegraphics[width=0.55\linewidth]{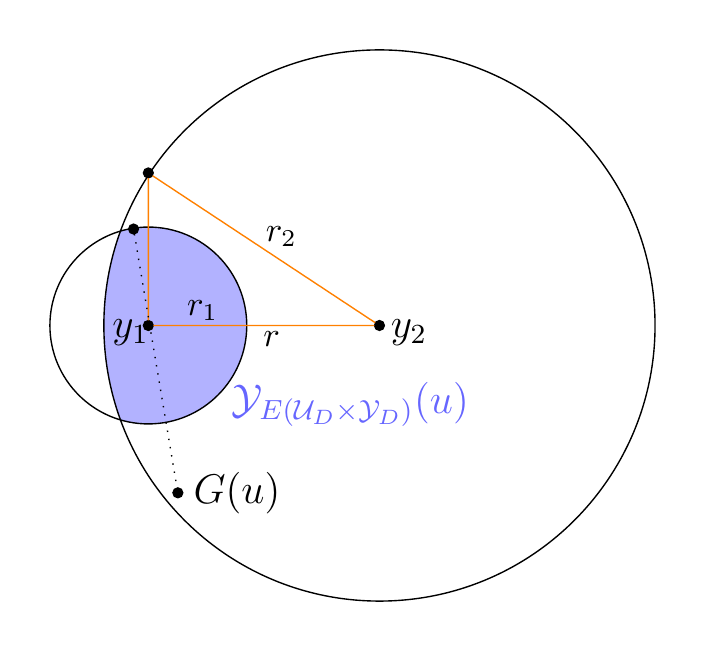}
	\end{center}
	\caption{Set of possible outputs $\mathcal{Y}_{E(\mathcal{U}_D^\ell\times\mathcal{Y}_D^\ell)}(u)$ by Lipschitz continuity at two samples. The case with center $M(u)$ not between the output samples $y_1$ and $y_2$.}
	\label{Fig:ProofCase2}
\end{figure}}
This case is characterized by $r_2(u)^2\geq r_1(u)^2+r(u)^2$ if $y_1$ lies in $\mathcal{Y}_{E(\mathcal{U}_D^\ell\times\mathcal{Y}_D^\ell)}(u)$ (by $r_1(u)^2\geq r_2(u)^2+r(u)^2$ if $y_2$ lies in $\mathcal{Y}_{E(\mathcal{U}_D^\ell\times\mathcal{Y}_D^\ell)}(u)$) as follows from the orange triangle in Figure~\ref{Fig:ProofCase2}. Since $y_1$ ($y_2$) lies in $\mathcal{Y}_{E(\mathcal{U}_D^\ell\times\mathcal{Y}_D^\ell)}(u)$,
\begin{equation*}
	\max_{y\in\mathcal{Y}_{E(\mathcal{U}_D\times\mathcal{Y}_D)}(u)}||y-G(u)||\leq ||y_{1(2)}-G(u)||+r_{1(2)}(u)
\end{equation*}
based on the triangle inequality. This results in $\beta_{\text{AE}}^{E(\mathcal{U}_D\times\mathcal{Y}_D),G}$ ($\gamma_{\text{AE}}^{E(\mathcal{U}_D\times\mathcal{Y}_D),G}$). 
\end{proof}
\DeleteSix{ 
\begin{thm}\label{Relaxation2}
Let two input-output samples be given $\mathcal{U}_D^\ell\times\mathcal{Y}_D^\ell=\{(u_1,y_1),(u_2,y_2)\}$. Then, the local AE-NLM inference \eqref{LocalOpt} is bounded from above by
\begin{equation}\label{NonconvexRelaxation2}
\begin{aligned}
	\max_{u\in \mathcal{U}_{H_i},||u||\geq\epsilon}\frac{R_1(u)+R_2(u)}{||u||}
\end{aligned}
\end{equation} 
with 
\begin{align*}
	R_{1(2)}(u)&=w_{1(2)}(u)(||y_{1(2)}-G(u)||+r_{1(2)}(u)),\\
	r_{1(2)}(u)&=L||u-u_{1(2)}||
\end{align*}
and the weightings $w_1(u)\geq0$ and $w_2(u)\geq0$, which satisfy $w_1(u)+w_2(u)=1$ for all $u\in\mathcal{U}_{H_i}$, e.g., $w_1(u)=w_2(u)=1/2$.
\end{thm} 
\begin{proof}
The triangular inequality immediately yields
\begin{align*}
	&\max_{y\in\mathcal{Y}_{E(\mathcal{U}_D^\ell\times\mathcal{Y}_D^\ell)}(u)}||y-G(u)||\\
	&\hspace{3cm}\leq \min_{i=1,2}\{||y_i-G(u)||+r_i(u)\}\\
	&\hspace{3cm}\leq w_1(u)(||y_1-G(u)||+r_1(u))\\
	&\hspace{3cm}\quad+w_2(u)(||y_2-G(u)||+r_2(u))
\end{align*}
for any weightings $w_i(u)\geq0, i=1,2$ with $w_1(u)+w_2(u)=1$. 
\end{proof}}
\DeleteSix{Although the relaxations from Theorem~\ref{Relaxation1} and \ref{Relaxation2} are nonconvex, the complexity of their optimization problems \eqref{NonconvexRelaxation1} and \eqref{NonconvexRelaxation2}}\Six{Although the relaxation from Theorem~\ref{Relaxation1} is nonconvex, the complexity of its optimization problems \eqref{NonconvexRelaxation1}} is significant lower compared to \eqref{LocalOpt} as the optimization over $y\in\mathbb{R}^n$ is avoided. Therefore, \DeleteSix{the relaxations \eqref{NonconvexRelaxation1} and \eqref{NonconvexRelaxation2} require}\Six{the relaxation \eqref{NonconvexRelaxation1} requires} the optimization of $\mu$ variables because the input set $\mathcal{U}$ is spanned by a $\mu$-dimensional orthonormal basis \eqref{InputSet}. \\\indent
So far the output trajectories of $\mathcal{Y}_D$ are assumed to be measured without noise. However, we can adapt the envelope $E(\mathcal{U}_D\times\mathcal{Y}_D)$ and the presented relaxation\DeleteSix{s} to provide a guaranteed upper bound on the AE-NLM for noisy measurements as shown in the next remark.
\begin{rmk}\label{RmkNoise2}
If the measured output $\tilde{y}$ of the system $N(u)$ is corrupted by additive and bounded noise $v$, i.e.,
\begin{equation*}
	\tilde{y}=N(u)+v,\quad v^Tv\leq\delta^2,
\end{equation*}
then the Lipschitz continuity implies
\begin{align*}
	||N(u')-\tilde{y}||\leq L||u'-u||+\delta.
\end{align*} 
Hence, we increase the radii $r_1(u)$ and $r_2(u)$ in \DeleteSix{relaxations \eqref{NonconvexRelaxation1} and \eqref{NonconvexRelaxation2}}\Six{Theorem~\ref{Relaxation1}} by $\delta$ to ensure a guaranteed upper bound on the AE-NLM. Analogously, if the noise exhibits a signal-to-noise-ration $\delta$, i.e., $v^Tv\leq\delta^2y^Ty$, then the Lipschitz continuity and the assumption $N(0)=0$ imply
\begin{equation*}
	||v||^2\leq\delta^2||N(u)||^2\leq\delta^2L^2||u||^2\leq\delta^2L^2||u_H||^2,
\end{equation*}
where $u_H$ denotes the largest input of the considered subset of the partition with respect to the Euclidean norm. Thereby, we increase the radii $r_1(u)$ and $r_2(u)$ by $\delta L||u_H||$.
\end{rmk}

\subsection{Iterative scheme for AE-NLM inference}\label{Sampling}

In the previous section, local inferences and \DeleteSix{two nonconvex relaxations are}\Six{a nonconvex relaxation are} studied to derive the (global) inference of the AE-NLM \eqref{UpperboundFixG} for given data samples. To improve the guaranteed upper bound, further experiments can be evaluated iteratively on the plant. We establish in the following such an iterative procedure similar to a branch-and-bound algorithm.
\begin{algo}[Iterative scheme for AE-NLM inference]\label{IterScheme}\indent
\begin{itemize}
	\item[1)] Suppose a set of input-output samples are given. Initially, compute the linear approximation model $G$ that minimizes the maximal distance to the data samples according to the semidefinite-program in \cite{c100}. Moreover, define a partition of $\bar{\mathcal{U}}$ by hyperrectangles $\bar{\mathcal{U}}_{H_1},\dots,\bar{\mathcal{U}}_{H_h}$ and compute the local AE-NLM inference \eqref{LocalOpt} and its maximizing\ \textquoteleft worst-case\textquoteright\ input for all $\bar{\mathcal{U}}_{H_i},i=1,\dots,h$. Set the number of iterations $k$ to zero.
	\item[2) ] Identify the hyperrectangle $\bar{\mathcal{U}}_{H_{i^*}}$ with the largest local AE-NLM inference ${\Phi_{\text{AE}}^{\mathcal{U},G*}}$ and add this to the sequence $\Phi_{\text{AE}}^{\mathcal{U},G}(0),\dots,\Phi_{\text{AE}}^{\mathcal{U},G}(k):=\Phi_{\text{AE}}^{\mathcal{U},G*}$. Moreover, add the corresponding\ \textquoteleft worst-case\textquoteright\ amplitudes $\bar{u}^*\in\bar{\mathcal{U}}_{H^*}$ to the sequence $\bar{u}^*(0),\dots,\bar{u}^*(k):=\bar{u}^*$.
	\item[3) ] Divide hyperrectangle $\bar{\mathcal{U}}_{H_{i^*}}$ into two hyperrectangles $\bar{\mathcal{U}}^1_{H_{i^*}}$ and $\bar{\mathcal{U}}^2_{H_{i^*}}$ by a $\mu-1$ dimensional hyperplane which is orthogonal to one dimension and divides the largest edge of $\bar{\mathcal{U}}_{H_{i^*}}$. Moreover, the hyperplane contains
	\begin{itemize}
		\item [3a) ] $\bar{u}^*$ if $1/\alpha<\text{vol}(\bar{\mathcal{U}}^1_{H_{i^*}})/\text{vol}(\bar{\mathcal{U}}^2_{H_{i^*}})<\alpha$ for some chosen $\alpha>0$;
		\item [3b) ] else the middle point of $\bar{\mathcal{U}}_{H_{i^*}}$.
	\end{itemize}
	\item[4) ] Determine the output of the plant for
	\begin{itemize}
		\item [4a) ] $\bar{u}^*$ in case of 3a);
		\item [4b) ] the middle point of $\bar{\mathcal{U}}_{H_{i^*}}$ in case of 3b).
	\end{itemize}
	\item[5) ] Compute the local inference of the AE-NLM for hyperrectangles $\bar{\mathcal{U}}^1_{H_{i^*}}$ and $\bar{\mathcal{U}}^2_{H_{i^*}}$. Saturate these local AE-NLM inferences by the local inference of $\bar{\mathcal{U}}_{H_{i^*}}$.
	\item[6) ] Set $\bar{\mathcal{U}}_{H_{i^*}}:=\bar{\mathcal{U}}^1_{H_{i^*}}$, $\bar{\mathcal{U}}_{H_{h+1}}=\bar{\mathcal{U}}^2_{H_{i^*}}$, $h:=h+1$, and $k:=k+1$. Go to Step 2). 	
	\end{itemize}		
\end{algo}
An illustration of Algorithm~\ref{IterScheme} is depicted in Figure~\ref{Fig:InputSampling}.
\begin{figure}
	\begin{center}
		\includegraphics[width=0.9\linewidth]{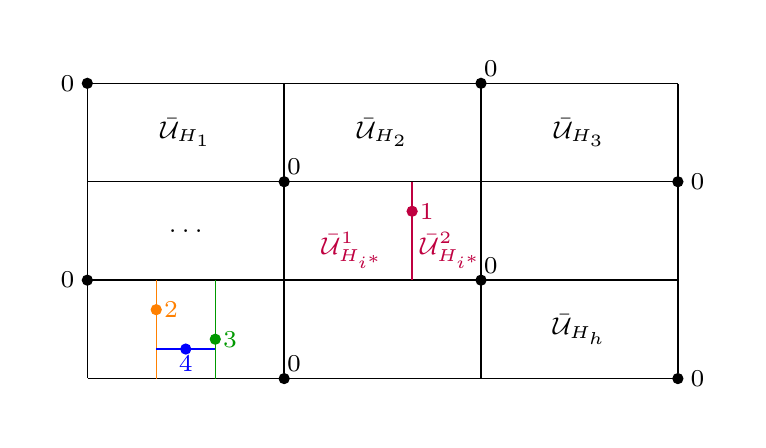}
	\end{center}
	\caption{Schematic illustration of the sampling procedure of the input space $\bar{\mathcal{U}}$ by Algorithm~\ref{IterScheme}. The numbers indicate the sampling in iteration $k$.}
	\label{Fig:InputSampling}
\end{figure}
In the sequel, we comment on Algorithm~\ref{IterScheme} and on some properties of Algorithm~\ref{IterScheme} more thoroughly.\\\indent
In step 3), we suggest one, among others, proceedings for dividing the hyperrectangle $\bar{\mathcal{U}}_{H^*}$. In particular, the decision between 3a) and 3b) is required to prove convergence of the sequence $\Phi_{\text{AE}}^{\mathcal{U},G}(0),\Phi_{\text{AE}}^{\mathcal{U},G}(1),\dots$ to the true AE-NLM in Theorem~\ref{ThmConv}. Furthermore, the new evaluated input from step 4) is taken into account in the computation of the local AE-NLM inferences in step 5) since its distance to the new hyperrectangles is zero. \\\indent
First property of Algorithm~\ref{IterScheme} is that its complexity does not increase with further iterations, as the local AE-NLM inference \eqref{LocalOpt} of two hyperrectangles is computed in each iteration. Hence, stopping the iteration after some iterations due to too large computation time won't occur once the algorithm can be initialized. Second, the sequence of global inferences $\Phi_{\text{AE}}^{\mathcal{U},G}(\cdot)$ from step 2) is monotone decreasing, as the local AE-NLM inferences of the hyperrectangles $\bar{\mathcal{U}}^1_{H_{i^*}}$ and $\bar{\mathcal{U}}^2_{H_{i^*}}$ are saturated by the local inference of $\bar{\mathcal{U}}_{H_{i^*}}$. Without saturation, the sequence could increase as the envelopes considered for $\bar{\mathcal{U}}^1_{H_{i^*}}$ and $\bar{\mathcal{U}}^2_{H_{i^*}}$, respectively, are not necessarily a subset of the envelope considered for $\bar{\mathcal{U}}_{H_{i^*}}$. Finally, the sequence of global inferences $\Phi_{\text{AE}}^{\mathcal{U},G}(\cdot)$ converges to the true AE-NLM as proven in the following theorem.
\begin{thm}[Convergence]\label{ThmConv}
The sequence of global inferences of the AE-NLM $\Phi_{\text{AE}}^{\mathcal{U},G}(\cdot)$ from step 2) of Algorithm~\ref{IterScheme} converges to the solution of the left-hand side of \eqref{Upperbound}, i.e.,
\begin{equation*}
	\lim_{k\rightarrow\infty}\Phi_{\text{AE}}^{\mathcal{U},G}(k)=\Phi^{\mathcal{U},G}_{\text{AE}}.
\end{equation*}
\end{thm}
\begin{proof}
The sequence $\Phi_{\text{AE}}^{\mathcal{U},G}(\cdot)$ is lower bounded by zero and non-increasing which implies its convergence. We suppose that the sequence $\Phi_{\text{AE}}^{\mathcal{U},G}(\cdot)$ doesn't converge to $\Phi^{\mathcal{U},G}_{\text{AE}}$. Thus, the sequence of corresponding\ \textquoteleft worst-case\textquoteright\ inputs ${u}^*(\cdot)\in\mathcal{U}$ from step 2) doesn't converge to the set of inputs ${\mathcal{U}}^*\subset{\mathcal{U}}$ which solve the left-hand side of \eqref{Upperbound}, but to a subset ${\mathcal{U}}_c\nsupseteq{\mathcal{U}}^*$. Due to the convergence to ${\mathcal{U}}_c$, the distinction of 3a) and 3b), and the sampling in step  4), we can choose ${\mathcal{U}}_c$ such that the radii $r_{1}(u)$ and $r_{2}(u)$ from relaxation\DeleteSix{s} \eqref{NonconvexRelaxation1}\DeleteSix{ and \eqref{NonconvexRelaxation2}} for each subset ${\mathcal{U}}_{H_i}\subseteq{\mathcal{U}}_c$ are bounded by an arbitrary small $\rho>0$, i.e.,
\begin{equation*}
	\max\{r_{1}(u),r_{2}(u)\}<\rho
\end{equation*} 
for all $u\in{\mathcal{U}}_{H_i}$ and all subsets ${\mathcal{U}}_{H_i}\subseteq{\mathcal{U}}_c$. Therefore, the local inference of the AE-NLM \eqref{LocalOpt} using the relaxation\DeleteSix{s} \eqref{NonconvexRelaxation1}\DeleteSix{ and \eqref{NonconvexRelaxation2}, respectively,} is bounded from above by
\begin{equation}\label{IneqProof}
	\max_{\substack{u\in \mathcal{U}_{H_i},||u||\geq\epsilon,\\i=1,2}}\frac{||y_i-G(u)||+\rho}{||u||}
\end{equation}
for each subset ${\mathcal{U}}_{H_i}\subseteq{\mathcal{U}}_c$.
The convergence
\begin{equation*}
	\lim_{\rho\rightarrow 0}\max_{\substack{u\in \mathcal{U}_{H_i},||u||\geq\epsilon,\\i=1,2}}\frac{||y_i-G(u)||+\rho}{||u||}=\frac{||y_1-G(u_1)||}{||u_1||}
\end{equation*}
implies that $\rho$ can be chosen small enough such that \eqref{IneqProof} is for all ${\mathcal{U}}_{H_i}\subseteq{\mathcal{U}}_c$ less than
\begin{equation*}
	\frac{||N({u}^*)-G(u^*)||}{||{u}^*||}, {u}^*\in{\mathcal{U}}^*.
\end{equation*}
Together with the search for the largest local AE-NLM inference in step 2) of Algorithm~\ref{IterScheme}, this leads to a contradiction for the convergence of the sequence ${u}^*(\cdot)\in\mathcal{U}$ to ${\mathcal{U}}_c\nsupseteq{\mathcal{U}}^*$. Hereby, the sequence ${u}^*(\cdot)$ converges to ${\mathcal{U}}^*$, and therefore $\Phi_{\text{AE}}^{\mathcal{U},G}(\cdot)$ to $\Phi^{\mathcal{U},G}_{\text{AE}}$.
\end{proof}

\DeleteSix{
\subsection{Extensions of Algorithm~\ref{IterScheme}}\label{ExtenScheme}

In this section, we present some modifications of Algorithm~\ref{IterScheme}. Remark~\ref{LocalLipschitz} and \ref{LocalInf} constitute two approaches to improve the estimation of the local inference~\eqref{LocalOpt}. 
\begin{rmk}\label{LocalLipschitz}
The consideration of local inferences of the AE-NLM ~\eqref{LocalOpt} in Algorithm~\ref{IterScheme} suggest to exploit local Lipschitz constants instead of a global Lipschitz constant. This leads to a less conservative uncertainty description of the envelope $E(\mathcal{U}_D\times\mathcal{Y}_D)$. For example, the online learning algorithms for Lipschitz constants in \cite{c102} yield estimations on the local Lipschitz constants.
\end{rmk}
\begin{rmk}\label{LocalInf}
Instead of analyzing the envelope of only one pair of samples for the computation of the local AE-NLM inferences~\eqref{LocalOpt}, we could compute \eqref{LocalOpt} for multiple pairs of samples and take their minimum as local inference. Moreover, instead of the saturation in step 5), we can guarantee a non-increasing sequence of AE-NLM inferences $\Phi_{\text{AE}}^{\mathcal{U},G}(\cdot)$ by calculating local inferences for pairs of samples until a smaller value as the local inference of the original hyperrectangle is obtained. 
\end{rmk}

In the initialisation of Algorithm~\ref{IterScheme}, we derive a linear approximation model of the nonlinear system from a given set of input-output trajectories. However, the original optimization \eqref{Upperbound} requires a combined optimization over the input set $\mathcal{U}$ and a set of linear models $\mathcal{G}$. For that reason, we include an iterative calculation of the linear model into Algorithm~\ref{IterScheme} in the next remark.
\begin{rmk}
To solve \eqref{Upperbound}, we adapt Algorithm~\ref{IterScheme} by an alternating optimization scheme similar to the alternating direction method of multipliers (ADMM) described in \cite{c107}. Thus, the iteration consists of two separated optimizations
\begin{align}	
	G^{k+1}=\underset{G\in\mathcal{G}}{\text{arginf}}\max_{(u_i,y_i)\in\mathcal{U}_D^k\times\mathcal{Y}_D^k}\frac{||y_i-G(u_i)||}{||u_i||}&,\label{ADMM1}\\
	\Phi^{E(\mathcal{U}_D^k\times\mathcal{Y}_D^k),G^{k+1}}_{\text{AE}}{=}\max_{\substack{(u,y)\in E(\mathcal{U}_D^k\times\mathcal{Y}_D^k),\\||u||\geq\epsilon}}\frac{||y{-}G^{k+1}(u)||}{||u||}&,\label{ADMM2}
\end{align}
where $\mathcal{U}_D^k\times\mathcal{Y}_D^k$ denotes the set of all samples after $k$ iterations following step 2)-4). Even though \eqref{ADMM1} can be solved efficiently by a semidefinite programming from \cite{c100}, the number of constraints increases with iterations. Moreover, the local AE-NLM inferences are required for all hyperrectangles to solve \eqref{ADMM2}. Thus, the complexity of the alternating optimization increases in each iteration. Thereby, this procedure is only tractable as an alternative to the initialization step 1). Furthermore, we cannot guarantee convergence of the sequence $G^0,G^1,\dots$ to the optimal linear model of \eqref{Upperbound}. 
\end{rmk}

\overlength{
Some optimization solvers require initially a point in the feasible set. To satisfy this for the optimization problems of relaxation from Theorem~\ref{Relaxation1}, we propose another sampling strategy as in step 4) in the last Remark~\ref{SamplingStrat}.

\begin{rmk}\label{SamplingStrat}
Instead of sampling once as in step 4), we could require two experiments in each iteration with inputs $\bar{u}_{H^*_{1}}\in\bar{\mathcal{U}}_{H^*_1}$ and $\bar{u}_{H^*_{2}}\in\bar{\mathcal{U}}_{H^*_2}$. These inputs are chosen such that the points at two opposite corners of hyperrectangle $\bar{\mathcal{U}}_{H^*_1}$ and $\bar{\mathcal{U}}_{H^*_2}$ are sampled. Note that the inputs of opposite corners of $\bar{\mathcal{U}}_{H^*}$, denoted by $\bar{u}_{H^*}'\in\bar{\mathcal{U}}_{H^*_1}$ and $\bar{u}_{H^*}''\in\bar{\mathcal{U}}_{H^*_2}$, are already sampled in some previous iteration step. Regarding the feasible sets of the optimization problems of \eqref{NonconvexRelaxation1} for $\bar{\mathcal{U}}_{H^*_1}$, either $\bar{u}_{H^*_{1}}=:u_1$ is a feasible point for $r_2^2>r_1^2+r^2$ or the feasible set is empty. Analogously, $\bar{u}_{H^*}'=:u_2$ is feasible for $r_1^2>r_2^2+r^2$ otherwise the feasible set is empty. Moreover, $1/2(u_1+u_2)$ is always a feasible point of $r_1^2-r_2^2\leq r^2$ and $r_2^2-r_1^2\leq r^2$. Furthermore, a reasonable choice of weightings for the relaxation \eqref{NonconvexRelaxation2} for $\bar{\mathcal{U}}_{H^*_{1}}$ is 
\begin{equation*}
	w_1(u)=1-\frac{||u-u_{H^*_{1}}||}{||u_{H^*_{1}}-u_{H^*}'||}, w_2(u)=\frac{||u-u_{H^*}'||}{||u_{H^*_{1}}-u_{H^*}'||}
\end{equation*}
with $w_1(u),w_2(u)\geq0$, as the basis of input signals is orthonormal, and hence $||u_{H^*_{1}}-u_{H^*}'||$ is the maximal distance of inputs in $\bar{\mathcal{U}}_{H^*_{1}}$. Analogical results hold for the feasible sets and weightings for $\bar{\mathcal{U}}_{H^*_{2}}$. 
\end{rmk}}

}

\section{Numerical Example}\label{SecEx}
 
In this section, we apply Algorithm~\ref{IterScheme} to conclude on the AE-NLM of the input-output behaviour of the SISO system 
\begin{align*}
	\dot{x}_1&=-3x_1+4x_2+2x_1^2-0.2\sin(3x_1)+u,\\
	\dot{x}_2&=-x_1+0.6x_2-0.5x_1^3,\quad x(0)=0,\\
	y&=x_1.
\end{align*}
The discrete-time input-output trajectories are drawn based on simulations for 30 time steps by Euler integration with $\Delta t=0.2\,\text{s}$. Similar to \cite{c100}, the input set is spanned by 
\begin{align*}
	\mathcal{U}{=}\{u\in\mathbb{R}^n:u{=}\sum_{i=1}^{2}\alpha_i \frac{v_i}{||v_i||}, (\alpha_1,\alpha_2)\in[0,4]^2\backslash[0,0.1]^2\},
\end{align*}
where $v_1$ and $v_2$ denote the stacked time-samplings of the basis signals $v_1(t)=\sin(\pi/3t)$ and $v_2(t)=\sin(\pi t)$. The Lipschitz constant is estimated to $L=1.04$. Algorithm~\ref{IterScheme} is initialized by the computation of a linear approximation model described by a lower Toeplitz matrix for 20 samples. The free parameter $\alpha$ of Algorithm~\ref{IterScheme} is chosen to $0.1$.\\\indent
Figure~\ref{Fig.NLM_iteration} shows the simulated sequence of AE-NLM inferences $\Phi_{\text{AE}}^{\mathcal{U},G}(\cdot)$ of Algorithm~\ref{IterScheme}\DeleteSix{ using the relaxation from Theorem~\ref{Relaxation1} and \ref{Relaxation2} with $w_1=w_2=1/2$}. 
\begin{figure}
	\centering
	\includegraphics[width=0.8\linewidth]{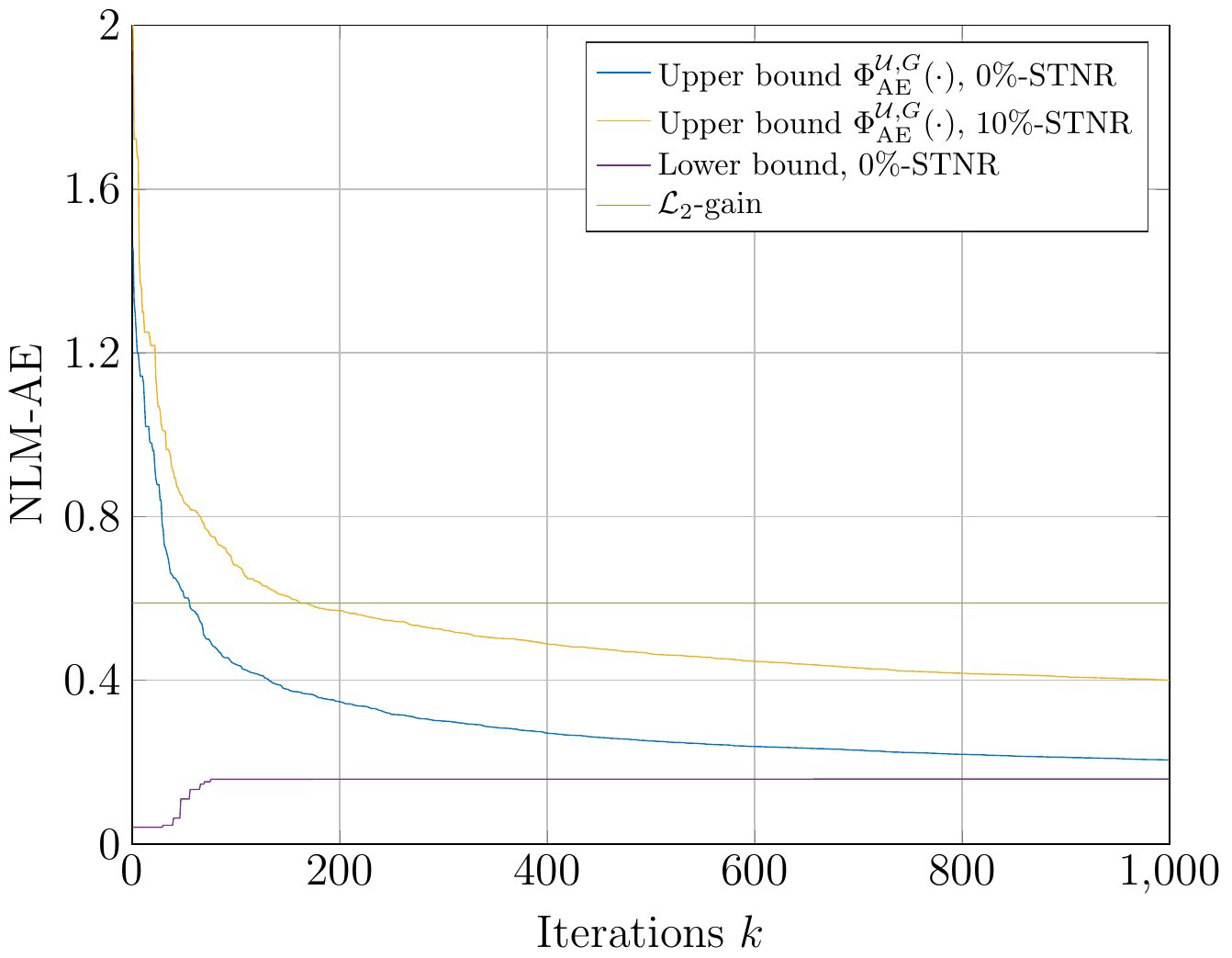}
	\caption{Simulated sequence of AE-NLM inference $\Phi_{\text{AE}}^{\mathcal{U},G}(\cdot)$ of Algorithm~\ref{IterScheme}.}
	\label{Fig.NLM_iteration}
\end{figure}
\DeleteSix{While Algorithm~\ref{IterScheme} exhibits for the relaxation from Theorem~\ref{Relaxation1} a slightly faster convergence rate, the computation of the local AE-NLM inferences for the relaxation from Theorem~\ref{Relaxation2} is slightly faster. }If the output is corrupted by additive noise with a signal-to-noise-ratio of $10\%$, then the adapted Algorithm~\ref{IterScheme} as described in Remark~\ref{RmkNoise2} still provides a guaranteed upper bound of the AE-NLM.\\\indent Furthermore, Figure~\ref{Fig.NLM_iteration} shows the lower bound from \cite{c100} on the AE-NLM derived from the collected data. The knowledge of the distance of the guaranteed lower and upper bound on the AE-NLM in each iteration constitutes a reasonable termination criterion for the iteration. Indeed, this distance measures the potential improvement of the AE-NLM estimation by further sampling.\\\indent Figure~\ref{Fig.NLM_input_space} demonstrates the division of the input set $\bar{\mathcal{U}}$ into hyperrectangles by Algorithm~\ref{IterScheme}.\DeleteSix{
\begin{figure}
	\centering
	\subfigure{\includegraphics[width=0.84\linewidth]{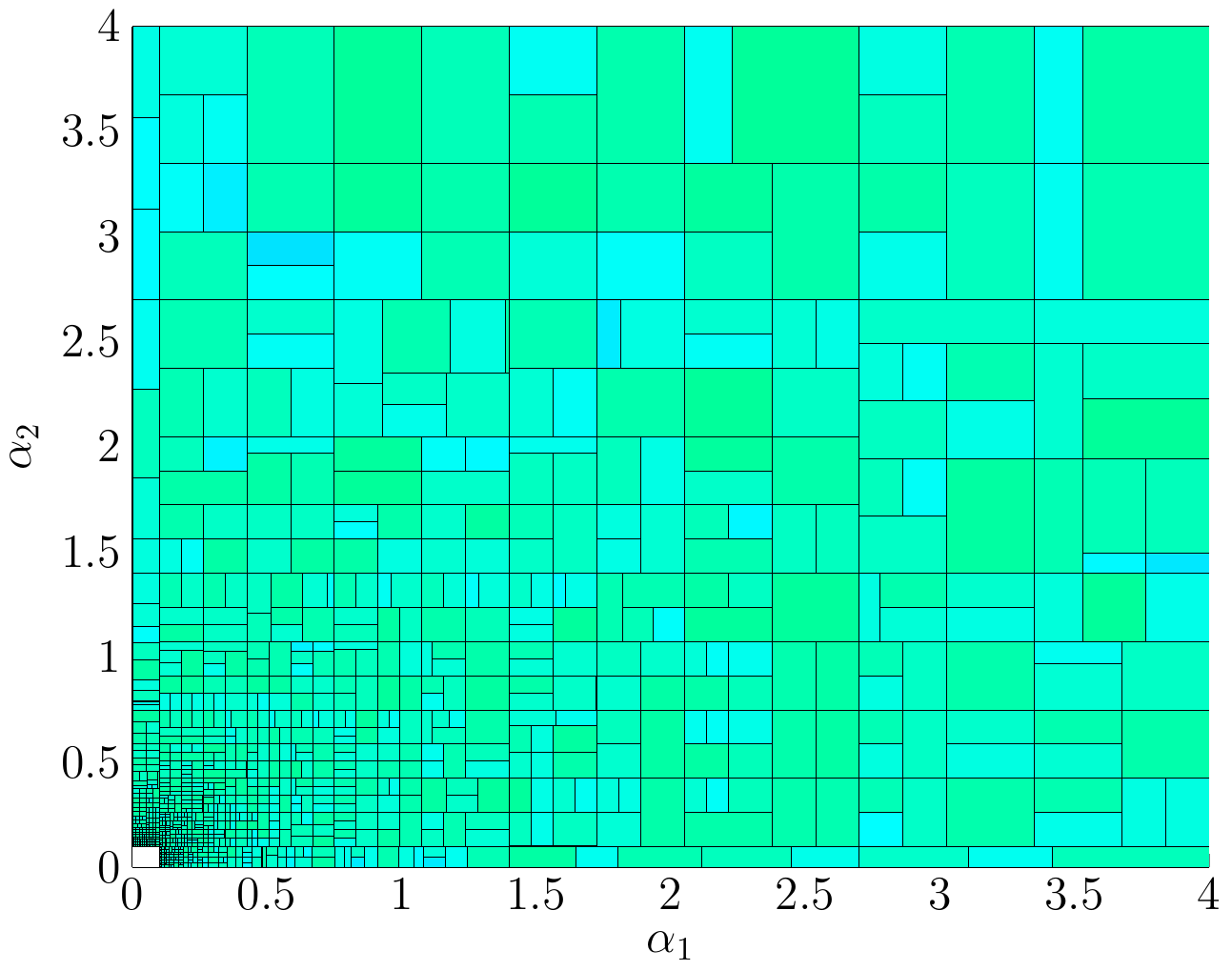}} \hspace{.01\textwidth}
	\subfigure{\includegraphics[width=0.07\linewidth]{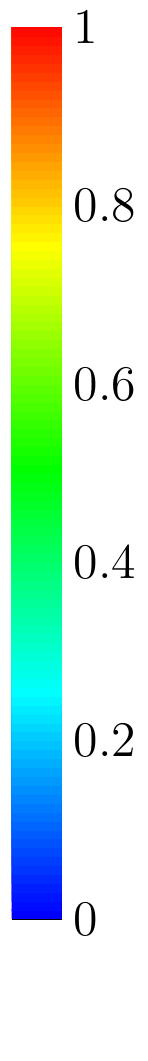}}
	\caption{Divided input set $\bar{\mathcal{U}}$ after $1000$ iterations\DeleteSix{ using the relaxation from Theorem~\ref{Relaxation1}}. The colour encodes the relation $\Phi_{\text{AE}}^{\mathcal{U}_{H_i},G}/\ell_2$ which is smaller than or equal to $1$ by definition. }
	\label{Fig.NLM_input_space}
\end{figure}}\Six{ 
\begin{figure}
	\centering
	\subfigure{\includegraphics[width=0.75\linewidth]{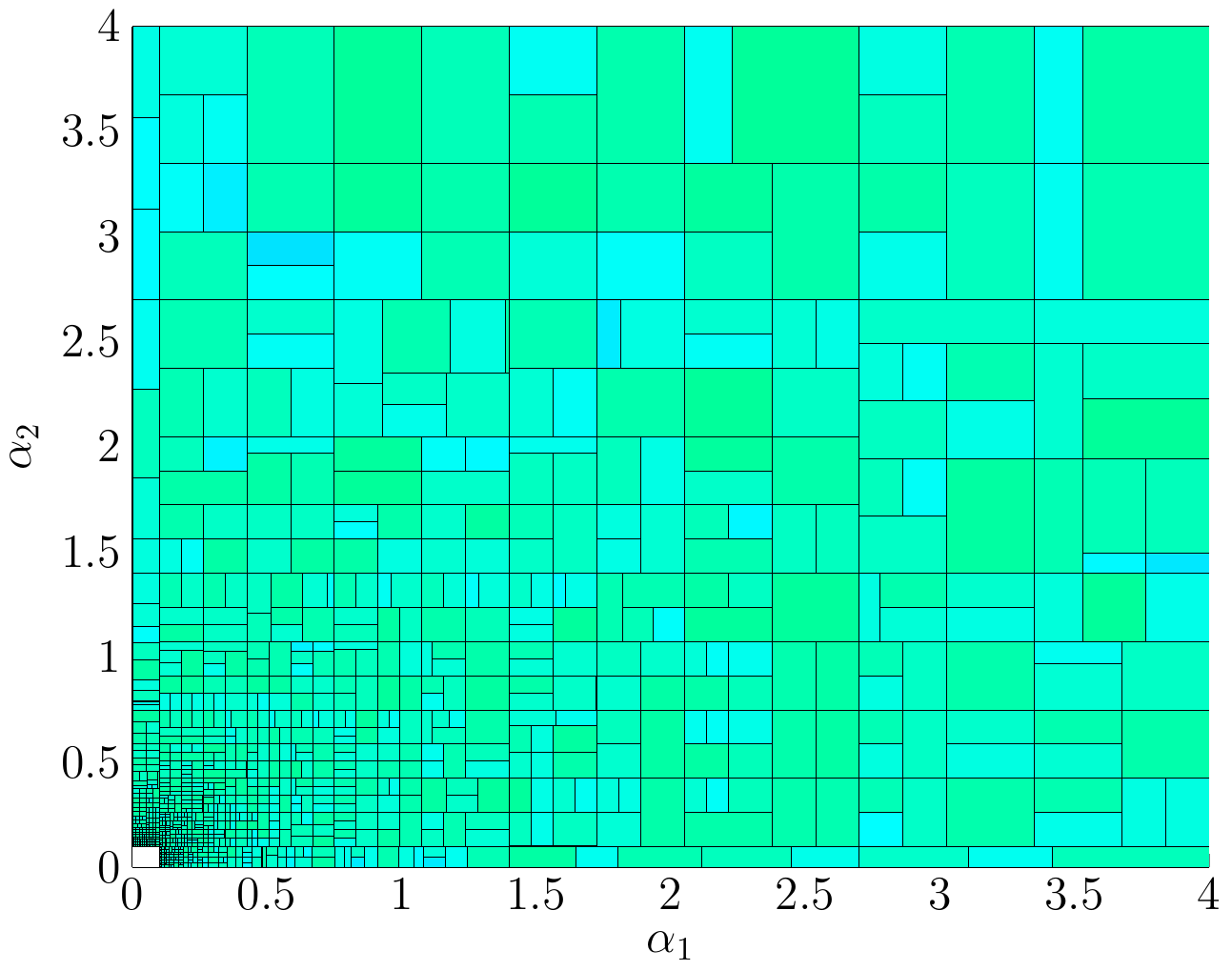}} \hspace{.01\textwidth}
	\subfigure{\includegraphics[width=0.07\linewidth]{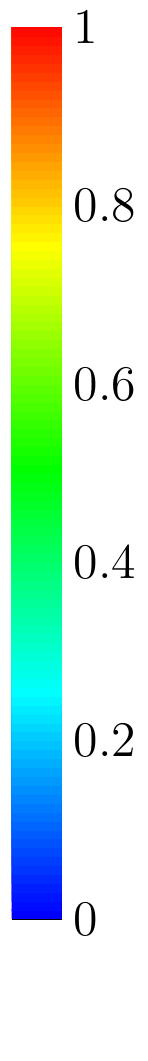}}
	\caption{Divided input set $\bar{\mathcal{U}}$ after $1000$ iterations\DeleteSix{ using the relaxation from Theorem~\ref{Relaxation1}}. The colour encodes the relation $\Phi_{\text{AE}}^{\mathcal{U}_{H_i},G}/\ell_2$ which is smaller than or equal to $1$ by definition. }
	\label{Fig.NLM_input_space}
\end{figure}} 
Due to the iterative approach, Algorithm~\ref{IterScheme} is more data-efficient than the offline approach \cite{c100} that requires $10000$ data samples to conclude on an upper bound of $0.48$ for the AE-NLM. Moreover, the computation time of Algorithm~\ref{IterScheme} for $1000$ iterations lasts around $5$ minutes which is comparable to \cite{c100}.

\section{Conclusions}

In this paper, we exploited a non-parametric data-based set-membership representation of the input-output mapping of an unknown nonlinear system to derive a conclusion on its strength of nonlinearity. First, we concluded on the nonlinearity measure from given input-output samples using local inference of the nonlinearity measure and a nonconvex \DeleteSix{two relaxations}\Six{relaxation}. Second, an iterative scheme was presented to decrease the guaranteed upper bound of the AE-NLM by further performed experiments. We ensured that the complexity of each iteration of the algorithm does not increase and proved the convergence to the true nonlinearity measure. \DeleteSix{Moreover, the algorithm was adapted, among others, by an alternating optimization of the linear approximation model. In a numerical example, the presented algorithm was more data efficient than the approach in \cite{c100}. In a future work, other dissipativity properties could be studied. Furthermore, the estimation could be improved if only causal mappings $N$ are considered, which is not necessarily required in our setup.}\Six{In a numerical example, the presented algorithm was more data efficient than the approach in \cite{c100}.\\\indent In a future work, other dissipativity properties could be studied, the iterative scheme could be extended by an alternating optimization \cite{c107} of the linear approximation model, and the estimation of the local inference could be improved by considering local Lipschitz constants.}

\end{document}